\documentclass[12pt]{article}
\usepackage{amsxtra}
\usepackage{amssymb, amsmath}
\usepackage{amsthm}
\usepackage {hyperref}

\newtheorem{theorem}{Theorem}[section]

\newtheorem{result}{Result}[section]
\newtheorem{proposition}{Proposition}[section]

\numberwithin{equation}{section}

\usepackage[dvips]{graphicx}

\begin{document}
\begin{center}
\textbf{\Large{Generalized Symmetric Divergence Measures and Metric Spaces}}
\end{center}

\bigskip
\begin{center}
\textbf{\large{G. A. T. F. da Costa }}\\
and\\
\smallskip
\textbf{\large{Inder Jeet Taneja}}\\
\smallskip
Departamento de Matem\'{a}tica\\
Universidade Federal de Santa Catarina\\
88.040-900 Florian\'{o}polis, SC, Brazil.\\
\end{center}

\begin{abstract}
Recently, Taneja \cite{tan3} studied two one parameter generalizations of \textit{J-divergence, Jensen-Shannon divergence } and \textit{Arithmetic-Geometric divergence}. These two generalizations in particular contain measures like: \textit{Hellinger discrimination}, \textit{symmetric chi-square divergence}, and \textit{triangular discrimination}. These measures are well known in the literature of Statistics and Information theory. In this paper our aim is to prove metric space properties for square root of these two symmetric generalized divergence measures.
\end{abstract}

\bigskip
\textbf{Key words:} \textit{J-divergence; Jensen-Shannon divergence; Arithmetic-Geometric divergence; Metric Space.}

\bigskip
\textbf{AMS Classification:} 94A17; 62B10.

\section{Introduction}

Let
\[
\Gamma _n = \left\{ {P = (p_1 ,p_2 ,...,p_n )\left| {p_i > 0,\sum\limits_{i
= 1}^n {p_i = 1} } \right.} \right\},
\quad
n \ge 2,
\]

\noindent
be the set of all complete finite discrete probability distributions. For all $P,Q \in \Gamma _n $, let us consider two generalized symmetric divergence measures. These measures are well known in the literature on information theory and statistics.

\bigskip
Let us consider the measure
\begin{equation}
\label{eq1}
\xi _s (P\vert \vert Q) =
\begin{cases}
 {L_s (P\vert \vert Q) = \left[ {s(s - 1)} \right]^{ - 1}\left[
{\sum\limits_{i = 1}^n {\left( {\frac{p_i^s + q_i^s }{2}} \right)\left(
{\frac{p_i + q_i }{2}} \right)} ^{1 - s} - 1} \right],} & {s \ne 0,1} \\
 {I(P\vert \vert Q) = \frac{1}{2}\left[ {\sum\limits_{i = 1}^n {p_i \ln
\left( {\frac{2p_i }{p_i + q_i }} \right) + \sum\limits_{i = 1}^n {q_i \ln
\left( {\frac{2q_i }{p_i + q_i }} \right)} } } \right],} & {s = 1} \\
 {T(P\vert \vert Q) = \sum\limits_{i = 1}^n {\left( {\frac{p_i + q_i }{2}}
\right)\ln \left( {\frac{p_i + q_i }{2\sqrt {p_i q_i } }} \right)} ,} & {s =
0} \\
\end{cases}
\end{equation}

\noindent
for all $P,Q \in \Gamma _n $

\bigskip
The measure (\ref{eq1}) was studied for the first time by Taneja \cite{tan2} and is called \textit{generalized symmetric} \textit{arithmetic and geometric mean divergence.} The measure (\ref{eq1}) admits the following particular cases:

\bigskip
\noindent
(i) $\xi _{ - 1} (P\vert \vert Q) = \frac{1}{4}\Delta (P\vert \vert Q)$.\\
(ii) $\xi _1 (P\vert \vert Q) = I(P\vert \vert Q)$.\\
(iii) $\xi _{1 / 2} (P\vert \vert Q) = 4\;d(P\vert \vert Q)$.\\
(iv) $\xi _0 (P\vert \vert Q) = T(P\vert \vert Q)$.\\
(v) $\xi _2 (P\vert \vert Q) = \frac{1}{16}\Psi (P\vert \vert Q)$.

\bigskip
\noindent where
\begin{align}
\Delta (P\vert \vert Q) & = \sum\limits_{i = 1}^n {\frac{(p_i - q_i )^2}{p_i +
q_i }} ,\notag\\
\intertext{and}
d(P\vert \vert Q) &= 1 - \sum\limits_{i = 1}^n {\left( {\frac{\sqrt {p_i } +
\sqrt {q_i } }{2}} \right)} \left( {\sqrt {\frac{p_i + q_i }{2}} } \right).\notag
\end{align}

\noindent
are the \textit{triangular discrimination} and \textit{d-divergence} respectively. The measures $I(P\vert \vert Q)$ and $T(P\vert \vert Q)$ are the well-known Jensen-Shannon divergence \cite{sib, bur} and Arithmetic-Geometric mean divergence \cite{tan2},  respectively.

\bigskip
Let us consider now the other measure
\begin{equation}
\label{eq2}
\zeta _s (P\vert \vert Q) = \begin{cases}
 {J_s (P\vert \vert Q) = \left[ {s(s - 1)} \right]^{ - 1}\left[
{\sum\limits_{i = 1}^n {\left( {p_i^s q_i^{1 - s} + p_i^{1 - s} q_i^s }
\right) - 2} } \right],} & {s \ne 0,1} \\
 {J(P\vert \vert Q) = \sum\limits_{i = 1}^n {\left( {p_i - q_i } \right)\ln
\left( {\frac{p_i }{q_i }} \right),} } & {s = 0,1} \\
\end{cases}
\end{equation}

\noindent
for all $P,Q \in \Gamma _n $

\bigskip
The measure (\ref{eq2}) can be seen in Burbea and Rao \cite{bur} and Taneja \cite{tan3}. The expression (\ref{eq2}) admits the following particular cases:

\bigskip
\noindent
(i) $\zeta _{ - 1} (P\vert \vert Q) = \zeta _2 (P\vert \vert Q) = \frac{1}{2}\Psi (P\vert \vert Q)$,\\
(ii) $\zeta _0 (P\vert \vert Q) = \zeta _1 (P\vert \vert Q) = J(P\vert \vert Q)$,\\
(iii) $\zeta _{1 / 2} (P\vert \vert Q) = 8\,h(P\vert \vert Q)$,
\bigskip

\noindent where
\begin{align}
\Psi (P\vert \vert Q) & = \chi ^2(P\vert \vert Q) + \chi ^2(Q\vert \vert P) =
\sum\limits_{i = 1}^n {\frac{(p_i - q_i )^2(p_i + q_i )}{p_i q_i }} ,\notag\\
\intertext{and}
h(P\vert \vert Q)&  = 1 - B(P\vert \vert Q) = \frac{1}{2}\sum\limits_{i = 1}^n
{(\sqrt {p_i } - \sqrt {q_i } )^2} ,
\end{align}

\noindent
are the \textit{symmetric }$\chi ^2$\textit{-- divergence} and \textit{Hellinger's discrimination} respectively. The measure $J(P\vert \vert Q)$ is the well-known J-divergence \cite{jef}. For detailed study of the measures (\ref{eq1}) and (\ref{eq2}) refer to Taneja \cite{tan2, tan3}.

\bigskip
The symmetric divergence measures (\ref{eq1}) and (\ref{eq2}) admit several particular cases. An inequality among these measures \cite{tan2} is given by
\begin{equation}
\label{eq3}
\frac{1}{4}\Delta (P\vert \vert Q) \le I(P\vert \vert Q) \le h(P\vert \vert
Q) \le 4\,d(P\vert \vert Q)
 \le \frac{1}{8}J(P\vert \vert Q) \le T(P\vert \vert Q) \le \frac{1}{16}\Psi
(P\vert \vert Q).
\end{equation}

An improvement over the inequalities given in (\ref{eq3}) can be seen in Taneja \cite{tan4, tan5}.

\bigskip
In this paper our aim is to prove metric space properties of the square root of the measures (\ref{eq1}) and (\ref{eq2}).

\section{Generalized Divergence Measures and Metric Spaces}

In this section we shall prove the metric space property of the square root of the measures given in (\ref{eq1}) and (\ref{eq2}).

\subsection{JS and AG -- Divergences of Type s}

Let the function $\zeta _s (p,q):{\rm R}^ + \times {\rm R}^ + \to {\rm R}^ + $be defined as
\begin{equation}
\label{eq4}
\zeta _s (p,q) = \begin{cases}
 {L_s (p,q) = \frac{1}{s(s - 1)}\left[ {\left( {\frac{p^s + q^s}{2}}
\right)\left( {\frac{p + q}{2}} \right)^{1 - s} - \left( {\frac{p + q}{2}}
\right)} \right],} & {s \ne 0,1} \\
 {I(p,q) = \frac{p}{2}\ln \left( {\frac{2p}{p + q}} \right) + \frac{q}{2}\ln
\left( {\frac{2q}{p + q}} \right),} & {s = 0} \\
 {T(p,q) = \left( {\frac{p + q}{2}} \right)\ln \left( {\frac{p + q}{2\sqrt
{pq} }} \right),} & {s = 1} \\
\end{cases}
\end{equation}

\noindent
In view of (\ref{eq4}), we can write
\begin{equation}
\label{eq5}
\zeta _s (P\vert \vert Q) = \sum\limits_{i = 1}^n {\zeta _s (p_i ,q_i )} ,
\end{equation}

\noindent
for all $P,Q \in \Gamma _n $

\begin{theorem} The measure given by $\sqrt {\zeta _s (p,q)} $ is a metric over ${\rm R}^ + $.
\end{theorem}

\begin{proof}  (i) Initially we shall prove the result for $s \ne 0,1$. It is sufficient to show the triangle inequality:
\begin{equation}
\label{eq6}
\sqrt {L_s (p,q)} \le \sqrt {L_s (p,r)} + \sqrt {L_s (r,q)} ,
\quad
\forall p,q,r \in {\rm R}^ + ,
\quad
s \ne 0,1
\end{equation}

\noindent
Let us write
\[
K_{pq} (r) = \sqrt {L_s (p,r)} + \sqrt {L_s (r,q)} .
\]

\noindent
Now we shall prove that $K_{pq} $ has only one minimum at $r = p = q$. The derivative o f$K_{pq} $ with respect to $r$ is
\[
{K}'_{pq} (r) = \frac{{L}'_s (p,r)}{2\sqrt {L_s (p,r)} } + \frac{{L}'_s
(r,q)}{2\sqrt {L_s (r,q)} },
\]

\noindent where
\begin{align}
{L}'_s (p,r) & = \frac{d}{dr}L_s (p,r)\notag\\
& = \frac{(1 - s)r^{ - s}(p + r)^s + (p^{1 - s} + r^{1 - s})s(p + r)^{s - 1}
- 2^s}{s(s - 1)2^{s + 1}}\notag\\
& \mathop = \limits_{\frac{p}{r} = t} \frac{(1 - s)(1 + t)^s + s(t^{1 - s} +
1)(t + 1)^{s - 1} - 2^s}{s(s - 1)2^{s + 1}}. \notag
\end{align}

\noindent
Also, we can write
\[
\sqrt {L_s (p,r)} \mathop = \limits_{\frac{p}{r} = t} \sqrt r \sqrt {L_s (t,1)} ,
\]

\noindent
Multiply $K'_{pq} $ by $2\sqrt r $ and define the function $h\left( t \right)$ by setting
\[
\frac{2\sqrt r {L}'_s (p,r)}{\sqrt {L_s (p,r)} }\mathop =
\limits_{\frac{p}{r} = t} h_{L_s } (t) = \frac{n_{L_s } (t)}{d_{L_s } (t)},
\]

\noindent where
\[
n_{L_s } (t) = \left. {\frac{d}{dr}L_s (p,r)} \right|_{\frac{p}{r} = t}
\]

\noindent and
\[
d_{L_s } (t) = \sqrt {L_s (t,1)} .
\]

\noindent
Thus the sign of $h_{L_s } (t)$ depends only on the sign of $n_{L_s } (t)$. We have
\[
{n}'_{L_s } (t) = - \frac{(1 + t)^{s - 2}}{2^{s + 1}}\left[ {(1 + t)(1 + t^{
- s}) + (1 + t^{1 - s})} \right].
\]

\noindent
This give
\bigskip
${n}'_{L_s } (t) < 0,\;\forall t > 0\;\mbox{and}\;\forall s  \Rightarrow n_{L_s } (t)$ decreases monotonically in $\left( {0, + \infty }
\right)$.

 \noindent
As $h_{L_s } (1) = 0$,  $n_{L_s } (t)$ changes the sign at $t = 1$, therefore $h_{L_s } (t)$ changes the sign at $t = 1$. This gives
\[
h_{L_s } (t)\begin{cases}
 { > 0,} & {t < 1} \\
 { < 0,} & {t > 1} \\
\end{cases}
\]

\noindent for any $s$.

\bigskip \noindent
As $\frac{p}{r} = t$, then $\frac{q}{r} = \frac{q}{p}\frac{p}{r} = \beta t$, where $\beta = \frac{q}{p}$. Therefore,
\[
2\sqrt r \frac{dK_{L_s } }{dr} = h_{L_s } (t) + h_{L_s } (t\beta ).
\]

\noindent
Now, suppose $\beta > 1\mbox{ }\left( {q > p} \right)$, this give:
\begin{itemize}
\item for $t < \frac{1}{\beta }$: $h_{L_s } (t)$ and $h_{L_s } (\beta t)$ have the same sign +
\item for $t > 1$: $h_{L_s } (t)$ and $h_{L_s } (\beta t)$ have the same sign --
\item for $t \in \left( {\frac{1}{\beta },1} \right) \Rightarrow t\beta > 1 \Rightarrow h_{L_s } (t\beta ) < 0$
\item for $t \in \left( {\frac{1}{\beta },1} \right) \Rightarrow h_{L_s } (t) > 0$
\end{itemize}

\bigskip \noindent
Finally, we have for $t \in \left( {\frac{1}{\beta },1} \right)$, $h_{L_s } (t) > 0$ e $h_{L_s } (t\beta ) < 0$. Since $\left| {h_{L_s } \left( t \right)} \right| > \left| {h_{L_s } \left( {t\beta } \right)} \right|$
($h_{L_s } $is monotonically decreasing), then $h_{L_s } (t) + h_{L_s } (\beta t) > 0$. For $t > 1$, $h_{L_s } (t) < 0$ e $h_{L_s } (t\beta ) < 0$ and $h_{L_s } (t) + h_{L_s } (\beta t) < 0$.

\bigskip \noindent
Therefore, $\frac{dK_{L_s } }{dr}$ indeed changes from positive to negative sign at $t = 1 \quad \left( {r = p} \right)$ so that there is a minimum at $t = 1$. Now, we shall show that this happens only once.

\bigskip \noindent
Since $h_{L_s } $ is monotonically decreasing this implies that ${h}'_{L_s } < 0$ and we know that the function $h_{L_s } $changes the sign only once.  This gives
\[
\frac{d}{dt}\left( {h_{L_s } (t) + h_{L_s } (t\beta )} \right) = {h}'_{L_s }
(t) + {h}'_{L_s } (t\beta ) < 0,
\]

\noindent
The case $q < p$ can be investigated in a similar fashion. Symmetry of $L_s $ allows us to take $t = \frac{q}{r}$ and $\frac{p}{r} = \beta t$ with $\beta = \frac{p}{q} > 1$. From this we conclude that there is a minimum at $r = q$.

\bigskip \noindent
Repeating the same process by substituting $t: = \frac{q}{r}$ and $\frac{p}{r} = \beta t$ with $\beta = \frac{p}{q}$ we conclude that the function ${K}'_{pq} $ also changes the sign at $r = q$. This proves the result (\ref{eq6}) for $s \ne 0,1.$

\bigskip
(ii) Now we shall prove the result for $s = 1$. We have to show that
\[
\sqrt {T(p,q)} \le \sqrt {T(p,r)} + \sqrt {T(r,q)} , \quad \forall p,q,r \in {\rm R}^ + .
\]

\noindent
Let us write
\[
T_{pq} (r) = \sqrt {T(p,r)} + \sqrt {T(r,q)} ,
\]

\noindent then obviously,
\[
{T}'_{pq} (r) = \frac{{T}'(p,r)}{2\sqrt {T(p,r)} } + \frac{{T}'(r,q)}{2\sqrt
{T(r,q)} }.
\]

\noindent
Now, we have
\begin{align}
{T}'(p,r) & = \frac{d}{dr}T(p,r)\notag\\
& = \frac{1}{2}\ln \left( {\frac{p + r}{2\sqrt {pr} }} \right) + \left(
{\frac{p + r}{2}} \right)\frac{d}{dr}\left[ {\ln \left( {\frac{p + r}{2\sqrt
{pr} }} \right)} \right]\notag\\
& = \frac{1}{2}\ln \left( {\frac{p + r}{2\sqrt {pr} }} \right) + \frac{r - p}{4r}. \notag
\end{align}

\noindent
This give
\[
\frac{{T}'(p,r)}{2\sqrt {T(p,r)} } = \frac{h_T \left( t \right)}{\sqrt {32r}
}\left| {_{t = \frac{p}{r}} } \right.
\]

\noindent where
\[
h_T (t): = \frac{2\ln \left( {\frac{t + 1}{2\sqrt t }} \right) + (1 -
t)}{\sqrt {(t + 1)\ln \left( {\frac{(1 + t)}{2\sqrt t }} \right)} }.
\]

\noindent
Let us take now $\frac{q}{r} = \beta t$ where $\beta = \frac{q}{p}$, then we can write
\[
\left. {\sqrt {32r} \,{T}'_{pq} (r)} \right|_{t = \frac{p}{r}} = h_T (t) +
h_T (\beta t).
\]

\noindent
Let us study now the function $h_T (t)$. Call $n_T (x)$ and $d_T \left( t \right)$ the functions in the numerator and denominator of $h_T (t)$, respectively. Since we know that $d_T (t) > 0,\;\forall x > 0$, then the sign of $h_T (t)$ is determined by $n_T (t)$.

\bigskip \noindent
Now,
\[
{n}'_T (t) = 2\frac{\sqrt t }{\left( {t + 1} \right)}\frac{d}{dt}\left(
{\frac{t + 1}{\sqrt t }} \right) - 1
 = - \frac{\left( {1 + t^2} \right)}{t(t + 1)}.
\]

\bigskip \noindent
From this we conclude that $n_T (t)$ is decreasing $\forall t > 0$ and ${n}'_T (t) \ne 0$, $\forall t \in {\rm R}$. Since $n_T (1) = 0$ then $n_T (t)$ changes the sign at $t = 1$. This gives that $h_T (t)$ changes the sign at $t = 1$.

\bigskip \noindent
Thus we conclude that $h_T (t)$ is decreasing function of $x$ and
\[
h_T (t)\begin{cases}
 { > 0,} & {t < 1} \\
 { < 0,} & {t > 1} \\
\end{cases}
\]

\noindent
Set $\beta > 1$. In this case,
\begin{itemize}
\item for $t < \frac{1}{\beta }$: $h_T (t)$ and $h_T (\beta t)$ have the same sign +
\item for $t > 1$: $h_T (t)$ and $h_T (\beta t)$ have the same sign --
\item for $t \in \left( {\frac{1}{\beta },1} \right) \Rightarrow t\beta > 1 \Rightarrow h_T (t\beta ) < 0$
\item for $t \in \left( {\frac{1}{\beta },1} \right) \Rightarrow h_T (t) > 0$
\end{itemize}

\bigskip \noindent
Finally, for $t \in \left( {\frac{1}{\beta },1} \right)$, $h_T (t) > 0$ e $h_T (t\beta ) < 0$. Thus we observe that the sign of $h_T (t) + h_T (\beta t)$ may change in $\left( {\textstyle{1 \over \beta },1} \right)$. Now, we shall show that this happens only once.

\bigskip \noindent
Since ${h}'_T $ is monotonically decreasing this implies that ${h}'_T < 0$ and we know that the function $h_T $changes the sign only once. This gives
\[
\frac{d}{dt}\left( {h_T (t) + h_T (t\beta )} \right) = {h}'_T (t) + {h}'_T (t\beta ) < 0,
\]

\noindent
Repeating the same process by substituting $t: = \frac{q}{r}$ and
$\frac{p}{r} = \beta t$ with $\beta = \frac{p}{q}$ we conclude that the
function ${T}'_{pq} $ also changes the sign at $r = q$.

\bigskip
(iii) For $s = 0$ the result is already given in Endres and Schindelin \cite{ens}.
\end{proof}

\subsection{J -- Divergences of Type s}

Let the function $\xi _s (p,q):{\rm R}^ + \times {\rm R}^ + \to {\rm R}^ + $be defined as
\begin{equation}
\label{eq7}
\xi _s (p,q) = \begin{cases}
 {J_s (p,q) = \left[ {s(s - 1)} \right]^{ - 1}\left[ {p^sq^{1 - s} + p^{1 -
s}q^s - (p + q)} \right],} & {s \ne 0,1} \\
 {J(p,q) = (p - q)\ln \left( {\frac{p}{q}} \right)} & {s = 0,1} \\
\end{cases}
\end{equation}

In view of (\ref{eq7}), we can write
\begin{equation}
\label{eq8}
\xi _s (P\vert \vert Q) = \sum\limits_{i = 1}^n {\xi _s (p_i ,q_i )} ,
\end{equation}

\noindent
for all $P,Q \in \Gamma _n $

\begin{theorem} The measure given by $\sqrt {\xi _s (p,q)} $ is a metric space over ${\rm R}^ + $.
\end{theorem}

\begin{proof} (i) Initially we shall prove the result for $s \ne 0,1$. It is sufficient to show the triangle inequality:
\begin{equation}
\label{eq9}
\sqrt {J_s (p,q)} \le \sqrt {J_s (p,r)} + \sqrt {J_s (r,q)} ,
\quad
\forall p,q,r \in {\rm R}^ + .
\end{equation}

\bigskip \noindent
Let us write
\[
F_{pq} (r) = \sqrt {J_s (p,r)} + \sqrt {J_s (r,q)} ,\;p \ne q,
\]

\noindent then obviously,
\[
{F}'_{pq} (r) = \frac{{J}'_s (p,r)}{2\sqrt {J_s (p,r)} } + \frac{{J}'_s
(r,q)}{2\sqrt {J_s (r,q)} }.
\]

\bigskip \noindent
Now, we have
\begin{align}
{J}'_s (p,r) & = \frac{d}{dr}J_s (p,r)\notag\\
& = \frac{(1 - s)p^sr^{ - s} + sp^{1 - s}r^{s - 1} - 1}{s(s - 1)}\notag\\
& \mathop = \limits_{\frac{p}{r} = t} \frac{t^sr + t^{1 - s}r - tr - r}{s(s - 1)}\notag
\end{align}

\bigskip \noindent
Also, we can write
\[
\sqrt {J_s (p,r)} \mathop = \limits_{p = rt} \sqrt r \sqrt {J_s (t,1)} ,
\]

\noindent
Let us write
\[
\sqrt r \frac{{J}'_s (p,r)}{\sqrt {J_s (t,1)} }\mathop =
\limits^{\frac{p}{r} = t} h_{J_s } (t) = \frac{n_{J_s } (t)}{d_{J_s } (t)},
\]

\noindent where
\[
n_{J_s } (t) = \left. {\frac{d}{dr}J_s (p,r)} \right|_{\frac{p}{r} = t}
\]

\noindent and
\[
d_{J_s } (t) = \sqrt {J_s (t,1)} .
\]

\noindent
Thus the sign of $h_{J_s } (t)$ depends on the sign of $n_{J_s } (t)$.
\[
{n}'_{J_s } (t) = - t^{s - 1} - t^{ - s}.
\]

\noindent
Thus
${n}'_{J_s } (t) < 0,\;\forall t > 0\;\mbox{and}\;\forall s \in {\rm R} - \{0,1\}  \Rightarrow n_{J_s } (t)$is decreasing $\forall t > 0$

\bigskip \noindent
As $h_{J_s } (1) = 0$, $n_{J_s } (t)$ changes the sign at $t = 1$ and therefore $h_{J_s } (t)$ changes the sign at $t = 1$. Thus for any $s$, we have
\[
h_{J_s } (t)\begin{cases}
 { > 0,} & {t < 1} \\
 { < 0,} & {t > 1} \\
\end{cases}
\]

\noindent
As $\frac{p}{r} = t$, then $\frac{q}{r} = \frac{q}{p}\frac{p}{r} = \beta t$, where $\beta = \frac{q}{p}$. Therefore,
\[
\sqrt r \frac{dJ_s }{dr} = h_{J_s } (t) + h_{J_s } (t\beta ).
\]

\noindent
Now,
\begin{itemize}
\item for $t < \frac{1}{\beta }$: $h_{J_s } (t)$ and $h_{J_s } (\beta t)$ have the same sign +
\item for $t > 1$: $h_{J_s } (t)$ and $h_{J_s } (\beta t)$ have the same sign --
\item for $t \in \left( {\frac{1}{\beta },1} \right) \Rightarrow t\beta > 1 \Rightarrow h_{J_s } (t\beta ) < 0$
\item for $t \in \left( {\frac{1}{\beta },1} \right) \Rightarrow h_{J_s } (t) > 0$
\end{itemize}

\noindent
Finally, for $t \in \left( {\frac{1}{\beta },1} \right)$, $h_{J_s } (t) > 0$ e $h_{J_s } (t\beta ) < 0$. Thus we observe that the sign of $h_{J_s } (t) + h_{J_s } (\beta t)$ may change in $\left( {\textstyle{1 \over \beta },1} \right)$. Now, we shall show that this happens only once.

\bigskip \noindent
Since ${h}'_{J_s } $ is monotonically decreasing this implies that  ${h}'_{J_s } < 0$ and we know that the function $h_{J_s } $changes the sign only once. This gives
\[
\frac{d}{dt}\left( {h_{J_s } (t) + h_{J_s } (t\beta )} \right) = {h}'_{J_s }
(t) + {h}'_{J_s } (t\beta ) < 0.
\]

\noindent
Repeating the same process by substituting $t: = \frac{q}{r}$ and $\frac{p}{r} = \beta t$ with $\beta = \frac{p}{q}$ we conclude that the function ${T}'_{pq} $ also changes the sign at $r = q$.

\bigskip \noindent
(ii) For $s = 0,1$, the result follows by the continuity of the function $\xi _s (p,q)$ with respect to $s$.
\end{proof}

\section{Asymptotic Approximation}

In this section we shall bring asymptotic approximation of the measures given by (\ref{eq1}) and (\ref{eq2}). For this, first we shall  give a general result for Csisz\'{a}r's $f$-divergence then the other cases become as particular.

\bigskip
Given a function$f:(0,\infty ) \to {\rm R}$, the \textit{f-divergence} measure introduced by Csisz\'{a}r's \cite{csi} is given by
\begin{equation}
\label{eq10}
C_f (P\vert \vert Q) =
\sum\limits_{i = 1}^n {q_i f\left( {\frac{p_i }{q_i }} \right)} ,
\end{equation}

\noindent
for all $P,Q \in \Gamma _n $.

\bigskip
The following result is well known in the literature \cite{csi}.

\begin{result}  If the function $f$ is convex and normalized, i.e., $f(1) = 0$, then the \textit{f-divergence}, $C_f (P\vert \vert Q)$ is \textit{nonnegative} and \textit{convex} in the pair of probability distribution $(P,Q) \in \Gamma _n \times \Gamma _n $.
\end{result}

Based on Result 3.1, we can prove some properties of the measures (\ref{eq1}) and (\ref{eq2}).
\begin{theorem}
If $f$ is twice differentiable at $x = 1$and ${f}''(1) > 0$. Also $f(1) = 0$, then
\begin{equation}
\label{eq11}
C_f (P\vert \vert Q) \approx \frac{{f}''(1)}{2}\chi ^2(P\vert \vert Q) .
\end{equation}

\noindent
Equivalently,
\[
\frac{C_f (P\vert \vert Q)}{\chi ^2(P\vert \vert Q)} \to
\frac{{f}''(1)}{2}\,\,\ \mbox{as}  \,\,P \to Q.
\]
\end{theorem}

\begin{proof} From Taylor's series expansion, we have
\[
f(x) = {f}'(1)(x - 1) + \frac{{f}''(1)}{2}(x - 1)^2 + k(x)(x - 1)^2,
\]

\noindent
where $f(1) = 0$ and $k(x) \to 0$ as $x \to 1$. Hence
\[
q_i f\left( {\frac{p_i }{q_i }} \right) = {f}'(1)(p_i - q_i ) +
\frac{{f}''(1)}{2}\frac{(p_i - q_i )^2}{q_i }
 + \,k\left( {\frac{p_i }{q_i }} \right)\frac{(p_i - q_i )^2}{q_i }.
\]

Approximating $p_i \to q_i $ and summing over all $i = 1,2,....n$ we get the required result.
\end{proof}

\begin{proposition} The following results hold:
\begin{itemize}
\item[(i)] $\zeta _s (P\vert \vert Q) \approx \frac{1}{8}\chi ^2(P\vert \vert Q)$,
$\forall s \in {\rm R}$.
\item[(ii)] $\xi _s (P\vert \vert Q) \approx \chi ^2(P\vert \vert Q)$, $\forall s
\in {\rm R}$.
\end{itemize}
\end{proposition}

\begin{proof} (i) For all $x > 0$ and $s \in ( - \infty ,\infty )$, let us consider
\begin{equation}
\label{eq12}
\psi _s (x) = \begin{cases}
 {\left[ {s(s - 1)} \right]^{ - 1}\left[ {\left( {\frac{x^{1 - s} + 1}{2}}
\right)\left( {\frac{x + 1}{2}} \right)^s - \left( {\frac{x + 1}{2}}
\right)} \right],} & {s \ne 0,1} \\
 {\frac{x}{2}\ln x - \left( {\frac{x + 1}{2}} \right)\ln \left( {\frac{x +
1}{2}} \right),} & {s = 0} \\
 {\left( {\frac{x + 1}{2}} \right)\ln \left( {\frac{x + 1}{2\sqrt x }}
\right),} & {s = 1} \\
\end{cases},
\end{equation}

\noindent
in (\ref{eq10}), then we have $C_f (P\vert \vert Q) = \zeta _s (P\vert \vert Q)$, where $\zeta _s (P\vert \vert Q)$ is as given by (\ref{eq1}).

\bigskip \noindent
We have
\[
\psi _s ^\prime (x) = \begin{cases}
 {(s - 1)^{ - 1}\left[ {\frac{1}{s}\left[ {\left( {\frac{x + 1}{2x}}
\right)^s - 1} \right] - \frac{x^{ - s} - 1}{4}\left( {\frac{x + 1}{2}}
\right)^{s - 1}} \right],} & {s \ne 0,1} \\
 { - \frac{1}{2}\ln \left( {\frac{x + 1}{2x}} \right),} & {s = 0} \\
 {1 - x^{ - 1} - \ln x - 2\ln \left( {\frac{2}{x + 1}} \right),} & {s = 1}
\\
\end{cases}
\]

\noindent and
\begin{equation}
\label{eq13}
\psi _s ^{\prime \prime }(x) = \left( {\frac{x^{ - s - 1} + 1}{8}}
\right)\left( {\frac{x + 1}{2}} \right)^{s - 2}.
\end{equation}

\noindent
This gives
\begin{equation}
\label{eq14}
\psi _s ^{\prime \prime }(1) = \frac{1}{4}.
\end{equation}

\noindent
Expression (\ref{eq11}) together with (\ref{eq13}) and (\ref{eq14}) completes the proof of part (i).

\bigskip \noindent
In particular, when $s = 0$ in (\ref{eq12}) the result is obtained in Endres and Schindelin \cite{ens}.

\bigskip
(ii) For all $x > 0$ and $s \in ( - \infty ,\infty )$, let us consider
\begin{equation}
\label{eq15}
\phi _s (x) = \begin{cases}
 {\left[ {s(s - 1)} \right]^{ - 1}\left[ {x^s + x^{1 - s} - (1 + x)}
\right],} & {s \ne 0,1} \\
 {(x - 1)\ln x,} & {s = 0,1} \\
\end{cases},
\end{equation}

\noindent
in (\ref{eq10}), then we have $C_f (P\vert \vert Q)=\xi _s \left( {P\vert \vert Q} \right)$, where $\xi _s \left( {P\vert \vert Q} \right)$ is given by (\ref{eq2}).

\bigskip \noindent
We have
\[
\phi _s ^\prime (x) = \begin{cases}
 {\left[ {s(s - 1)} \right]^{ - 1}\left[ {s(x^{s - 1} + x^{ - s}) + x^{ - s}
- 1} \right],} & {s \ne 0,1} \\
 {1 - x^{ - 1} + \ln x,} & {s = 0,1} \\
\end{cases},
\]

\noindent and
\begin{equation}
\label{eq16}
\phi _s ^{\prime \prime }(x) = x^{s - 2} + x^{ - s - 1}.
\end{equation}

\noindent
This gives
\begin{equation}
\label{eq17}
\phi _s ^{\prime \prime }(1) = 2.
\end{equation}

\noindent
Expression (\ref{eq17}) together with (\ref{eq10}) and (\ref{eq11}) completes the proof of part (ii).
\end{proof}

\end{document}